\documentclass[letterpaper, 10 pt, journal,twoside]{IEEEtran}

\usepackage{amsmath,amsfonts}
\usepackage{algorithmic}
\usepackage{array}
\usepackage[caption=false,font=normalsize,labelfont=sf,textfont=sf]{subfig}
\usepackage{textcomp}
\usepackage{stfloats}
\usepackage{url}
\usepackage{verbatim}
\usepackage{graphicx}
\def\BibTeX{{\rm B\kern-.05em{\sc i\kern-.025em b}\kern-.08em
    T\kern-.1667em\lower.7ex\hbox{E}\kern-.125emX}}
\usepackage{balance}

\usepackage{multirow}
\usepackage{booktabs}

\usepackage{graphics}
\usepackage{epsfig} 
\usepackage{xcolor}
\usepackage{amssymb}  
\usepackage{amsthm}
\usepackage{euscript}
\usepackage{xcolor}
\usepackage{tasks}
\usepackage{hyperref}
\usepackage{paralist}
\usepackage{comment}
\usepackage{amsmath} 
{
\newtheorem{theorem}{\bf Theorem}
\newtheorem{lemma}{\bf Lemma}

\newtheorem{problem}{\bf Problem}

\newtheorem{assumption}{\bf Assumption}

}

\newcommand{\trans}{^{\top}}

\newcommand{\vect}[1]{\boldsymbol{#1}}
\newcommand{\bvect}[1]{\Bar{\boldsymbol{#1}}}

\begin{document}

\title{\LARGE \bf Safe and robust tube-based path-following for robot navigation}

\author{Arthur H. D. Nunes, Vinicius M. Gon\c{c}alves, Guilherme V. Raffo,\\Leonardo A. B. Torres, and Luciano C. A. Pimenta%
\thanks{The authors are with the Graduate Program in Electrial Engineering (PPGEE) of the Universidade Federal de Minas Gerais (UFMG), Belo Horizonte, MG, 31270-901, Brazil (emails: arthurhdn@ufmg.br; vinicius.marianog@gmail.com; raffo@ufmg.br; leotorres@ufmg.br; lucpim@ufmg.br). The first author is affiliated with NUAR Systems LTDA, Contagem, MG, Brazil.}
\thanks{This work was in part supported by the grants CNPq (Brazilian National Research Council) 309925/2023-1 and 301298/2025-4, and FAPEMIG (under grants APQ-02144-18 and APQ-0063023). This work was partially supported by Petrobras/ANP under grants 2023/00494-5 and 2023/00643-0.}%
}


\maketitle

\begin{abstract}
In this paper, we propose a new robust navigation framework for path following tasks in robots operating within unknown, cluttered environments. Our approach ensures reactive safety through obstacle avoidance and guaranteed convergence to a target path, while simultaneously mitigating the impact of unknown-but-bounded disturbances using a tube-based control strategy. The methodology integrates key aspects in the robot navigation: (i) a nominal Integrated Guidance and Control scheme for path-following employing Artificial Vector Fields guidance and Backstepping control; (ii) a smooth distance function that enables a continuous control law formulation for seamless obstacle avoidance; (iii) a unified control objective that balances collision avoidance with path-following; and (iv) an adaptive control component to provide robustness against external disturbances. We provide formal proofs of safety and stability using barrier functions and Lyapunov stability theory. The effectiveness of the proposed framework is validated through extensive numerical simulations.
%
\end{abstract}


\section{Introduction}
\IEEEPARstart{A}{utonomous} robot navigation in unknown and dynamic environments is a timely and important challenge in robotics and control \cite{DurrantWhyte1994,Siciliano2016,11003579}. This task requires controllers capable of ensuring safety, robustness, and accurate path-following. 
The solution to the autonomous navigation problem can be applied to a wide range of applications such as surveillance, inspection, search
and rescue, urban mobility, autonomous vehicles, among others \cite{8384677,8714418,miranda2022autonomous,rezende2019autonomous}. 

\subsection{Obstacle avoidance}
The presence of obstacles is an important challenge in robot navigation. 
Strategies to address this issue are broadly categorized into two types: reactive and deliberative. 

Reactive approaches typically assume an obstacle-free workspace and perform online evasive maneuvers upon detection via exteroceptive sensors. For instance, 
\cite{yao2019integrated} have presented an integrated path-following approach based on the work of \cite{gonccalves2010vector}, and reactive collision avoidance through bump functions. Similarly, \cite{pereira2021collision} have proposed a reactive scheme for fixed-wing UAVs to deviate from moving obstacles. More recently, \cite{nunes2022vector} have extended these capabilities to avoid generic-shaped dynamic obstacles in $n$-dimensional spaces.
In contrast, deliberative approaches rely on offline path planning to compute a collision-free path. While deliberative methods can prevent local minima and generate smoother paths, they generally incur a higher computational cost and require prior knowledge of the map and obstacle geometries \cite{nunes2022vector}.

In general, reactive approaches are favored for their simplicity, speed, and low computational cost in unknown environments. However, they often suffer from non-smooth control laws and potential convergence to undesirable local equilibria. The non-smoothness of the reactive approaches often arises from the gradient of the standard Euclidean distance function. To mitigate this, smooth distance functions, such as those presented in \cite{11358638}, offer an effective alternative. The recent work \cite{nunes2026safe} expanded upon the methodology in \cite{nunes2022vector} by incorporating a smooth distance function and adding key features to provide safety guarantees.

In this work, we propose a reactive avoidance framework that uses the smooth distance function from \cite{nunes2026safe}. Unlike previous implementations that act as velocity-level commands, our approach integrates the function into an acceleration-based controller. 

\subsection{Guidance laws}

To address part of the navigation problem, many works have proposed guidance strategies based on Artificial Vector Fields (AVFs)  \cite{yao2021singularity,gonccalves2010vector,gonccalves2010circulation,yao2022guiding,yao2019integrated,lawrence2008lyapunov,nelson2007vector,rezende2020robust,pereira2021collision}.
%
%
The method proposed in \cite{rezende2021constructive} has introduced an approach for time-varying path-following in $n$-dimensions. 
By using a parametric representation of the target curve, this framework offers several advantages over traditional methods, such as those in \cite{yao2021singularity,gonccalves2010vector,rezende2020robust,pereira2021collision}. A key feature of the \cite{rezende2021constructive} method is the use of the closest point on the path to compute the guidance vector field. In the present work, we adopt this method directly as a nominal guidance layer to follow the path. We use the guidance law \cite{rezende2021constructive} as a constructive part of our methodology while incorporating other control aspects, such as obstacle avoidance, acceleration control, and disturbance mitigation.

\subsection{Control laws}

While guidance and obstacle avoidance are fundamental to robot navigation, they are typically implemented as an outer loop within hierarchical, cascaded approaches \cite{pereira2021collision,rezende2020robust}. In such schemes, the guidance layer is decoupled from the inner-loop controller (SGC - Separated Guidance and Controller). In contrast, we propose an Integrated Guidance and Controller (IGC) scheme. By leveraging Backstepping, a nonlinear control technique \cite{khalil2002nonlinear}, we unify the guidance law from \cite{rezende2021constructive} directly into the controller design. 

Backstepping offers significant advantages over linear control strategies, such as PID or LQR \cite{rubi2020survey,ozbek2016feedback,bouabdallah2004pid,cowling2007prototype}, particularly in handling the inherent nonlinearities of robotic systems \cite{salierno2018whole,TAN201647,5428769,5531424,1570447}.
While it has been successfully applied to various robotic platforms \cite{salierno2018whole, TAN201647}, we extend its utility here by employing it as a comprehensive path controller that incorporates both nominal guidance and safety constraints. Building upon the preliminary IGC concepts introduced in \cite{nunes2023integrated}, the present methodology evolves the framework by integrating formal safety guarantees based on \cite{nunes2026safe} and adaptive disturbance mitigation. This results in a unified acceleration-level control law that simultaneously addresses path convergence, obstacle avoidance, and robust performance under uncertainty.


\subsection{Robustness}

Controllers designed only for nominal operating conditions often exhibit significant performance degradation when subjected to disturbances. To address this, advanced control schemes aim to integrate nominal performance with disturbance attenuation capabilities by combining nonlinear and robust control techniques. The main objective is to design controllers that guarantees asymptotic stability under nominal conditions while effectively mitigating the effects of disturbances \cite{Mayne2005, 4605311, 6639914}. While these existing strategies provide a foundation for robustness, they typically do not utilize a Backstepping framework for the integrated design.

In a recent work \cite{nunes2024robust}, a robust adaptive Backstepping attitude controller based on \cite{Zhang2021} was developed. That scheme has introduced elements to solve the trajectory tracking problem, including a parameter projection modification \cite{farrell2006adaptive} within the adaptive gains to prevent them to drift. By combining a nominal controller with an adaptive component, the system achieved robustness against uncertainties while tracking a time-varying references. 
%
%

The present work adopts the robustness idea of \cite{nunes2024robust} but extends its application to $n$-dimensional path-following and reactive obstacle avoidance.

\subsection{Contributions}

In summary, the proposed methodology extends the foundational concepts presented in\cite{nunes2023integrated,nunes2024robust,nunes2026safe}. The primary contributions of this work are as follows:
\begin{itemize}
    \item The proposition of a novel nominal acceleration control law capable of following dynamic paths and reactively avoiding generic-shaped dynamic obstacles. This formulation advances the work in \cite{nunes2023integrated} by introducing a safety framework that allows formal proofs;
    \item The extension of a nominal IGC scheme that unifies AVF with Backstepping control \cite{rezende2021constructive,nunes2023integrated}, ensuring cohesive interaction between high-level guidance and vehicle dynamics; 
    \item The adaptation of a smooth, Euclidean-like distance function to achieve a continuous acceleration control law  \cite{nunes2026safe};
    \item The development of a combined nominal and robust adaptive control law. This extends the disturbance mitigation strategy of \cite{nunes2024robust} from attitude tracking to the $n$-dimensional safe path-following problem;
    \item Formal proofs for safety, path-following, and disturbance mitigation using barrier function concepts and Lyapunov stability theory;
    \item Extensive validation in simulated scenarios.
\end{itemize}

\section{Problem statement}

Let $\vect{p} \in \mathbb{R}^n$ be the system position, $\vect{v} \in \mathbb{R}^n$ be its velocity and $\vect{a} \in \mathbb{R}^n$ be its acceleration, also assumed as the control input. The system is considered disturbed by an acceleration term $\vect{a}_d \in \mathbb{R}^n$. The following equations describe its dynamics:
\begin{align}
    \Dot{\vect{p}} & = \vect{v},\nonumber \\
    \Dot{\vect{v}} & = \vect{a} + \vect{a}_d. \label{eq:system_double_integrator}
\end{align}
This system is said to be the \textbf{practical system}.

The primary objective of this work is to control the positions of the practical system to follow a small tube around a pre-defined target curve as formalized below. 
\begin{problem}(Path-following)
    Consider the disturbed system \eqref{eq:system_double_integrator} and a given time-varying target path $\vect{\mathcal{C}}(t)$ to be followed with parametric representation $\vect{c}(s,t)$. Design a control law that determines the acceleration $\vect{a}(\vect{p},\vect{v},t)$, such that the trajectories of the system, $\vect{p}(t)$, converge to and follow a small tube with known radius $R \in \mathbb{R}$ around the target path $\vect{\mathcal{C}}(t)$ when not obstructed by obstacles.
    \label{problem:path_following2}
\end{problem}

When the path is obstructed or a collision is at risk, the controller must prioritize safety over the path following. This means that the robot is allowed to move away from the curve to accomplish the obstacle deviation and then resume the primary goal once the avoidance is finished. 

To define the safety problem to be solved, let $\vect{\mathcal{O}}(t) \subset \mathbb{R}^n$ be a forbidden set, also referred to as \emph{obstacle set}. 
Assume that this set is represented by a finite number of time-varying points \(\vect{o}_i(t) \in \vect{\mathcal{O}}(t)\), in which the number of points stays the same, \(N\), and each function \(\vect{o}_i(t)\) is differentiable in \(t\).
Let the half-squared distance between the point $\vect{p}$ and the set $\vect{\mathcal{O}}(t)$ be defined as
\begin{equation}
D_{\vect{\mathcal{O}}}(\vect{p},t) \equiv \min_{i=1,...,N}\frac{1}{2}\|\vect{p} - \vect{o}_i(t)\|^2,
\label{eq:half_squared_distance}
\end{equation}
where $\vect{o}_i(t) \in \vect{\mathcal{O}}(t)$ is a time-varying point whose position is differentiable with respect to time.
%
Then, this problem is defined as follows.
\begin{problem}(Safety)
    The robot must avoid the forbidden set by keeping at least a safe distance $\lambda>0$ from it, \emph{i.e.}
    \begin{align}
        \vect{p}(t) & \in \vect{\mathcal{S}}(t) \quad \forall t \geq 0,
        \label{eq:condition_collision_avoidance}
    \end{align}
   where $\vect{\mathcal{S}}(t)$ is a safety set defined using barrier certificates \cite{8796030}:
    \begin{align}
        \vect{\mathcal{S}}(t) = & \{ \vect{p} \in \mathbb{R}^n \ | \ B^\lambda(\vect{p},t) \geq 0\},\label{eq:safety_set}\\
        B^\lambda(\vect{p},t) & \equiv D_{\vect{\mathcal{O}}}(\vect{p},t) - \frac{\lambda^2}{2}.\label{eq:safety_set_barrier_function}
    \end{align}
    \label{problem:collision_avoidance}
\end{problem}

The compromise idea is to always solve Problem \ref{problem:collision_avoidance} and, when it is safe, solve Problem \ref{problem:path_following2}. 

To address this scenario, our solution is composed of three main parts.
\begin{itemize}
    \item First, we assume a nominal system without disturbances, denoted by a top bar $(\Bar{.})$. In this system, we propose a nominal path-following control law $\Bar{\vect{a}}_{\vect{\Phi}}$ that achieves asymptotic convergence;
    \item Then, we propose a nominal avoidance control law $\Bar{\vect{a}}_{\vect{\Psi}}$ and blend it with the previous one ensuring nominal safety and path convergence;
    \item Finally, we propose a robust adaptive control scheme combined with the previous nominal control to ensure the deviation between the practical system and the previously found nominal trajectory remains small in a tube-based approach.
\end{itemize}

\section{Nominal control}

For the first part of our method, assume that there are no disturbances and let the top bar $(\Bar{.})$ denote the \textbf{nominal system} according to
\begin{align}
    \Dot{\Bar{\vect{p}}} & = \Bar{\vect{v}},\nonumber \\
    \Dot{\Bar{\vect{v}}} & = \Bar{\vect{a}}. \label{eq:system_double_integrator_nominal}
\end{align}

\subsection{Path-following}

In order to ensure only path-following in an obstacle-free scenario for now, we adopt a Backstepping approach where we integrate the Artificial Vector Field guidance from \cite{rezende2021constructive} into the controller.

\subsubsection{Step 1}
Following this strategy, in step one we consider $\Dot{\Bar{\vect{p}}} = \Bar{\vect{v}}$, where a virtual control law is designed for $\Bar{\vect{v}}$. This virtual control law follows the AVF procedure from \cite{rezende2021constructive} as a constructive block part of our solution. 

The computation of the vector field virtual control law for $\Bar{\vect{v}}$, defined as $\vect{\Phi}(\Bar{\vect{p}},t)$, starts by finding the closest point in the curve the respective distance, and tangent vectors
\begin{align}
    s_*(\bvect{p},t) & = \arg \min_s \| \bvect{p} - \vect{c}(s,t) \|,\label{eq:s_star} \\
     \vect{c}_*(\bvect{p},t) & = \vect{c}(s_*(\bvect{p},t),t),\label{eq:c_star}\\
     \vect{D}_{\vect{\mathcal{C}}}(\bvect{p},t) & = \bvect{p} - \vect{c}_*(\bvect{p},t), \label{eq:curve_distance_vector} \\
     \vect{T}_{\vect{\mathcal{C}}}(\bvect{p},t) & = \dfrac{\partial \vect{c}(s,t)}{\partial s}\biggr\rvert_{s = s_*}.\label{eq:field_T}
\end{align}

According to the method \cite{rezende2021constructive}, we assume that the domain of the function $\vect{\Phi}$ is the set of points $(\Bar{\vect{p}},t)$ such that $s_*(\bvect{p},t)$ is a singleton. Thus, the subsequent entities are also singletons. This assumption is relevant for properly defining the vector field. In practice, the set of points where $s_*(\bvect{p},t)$ is multi-valued has measure zero. In some cases, it is also possible to tune the vector field such that those points are unapproachable. In a general case, one can simply choose one of the possible vectors.

The two vectors in \eqref{eq:curve_distance_vector} and \eqref{eq:field_T} that represent the convergence and circulation behaviors respectively are weighted by scalar gains $G(\|\vect{D}_{\vect{\mathcal{C}}}\|) = \frac{2}{\pi}\tan^{-1}(k_G\|\vect{D}_{\vect{\mathcal{C}}}\|), \ k_G > 0,$ and $H(\|\vect{D}_{\vect{\mathcal{C}}}\|) = \pm\sqrt{1 - G(\|\vect{D}_{\vect{\mathcal{C}}}\|)^2}$.

Then, we compute a static unit field and a feed-forward component to deal with time variation:
\begin{align}
        \vect{\Phi}_S & = -G\frac{\vect{D}_{\vect{\mathcal{C}}}}{\|\vect{D}_{\vect{\mathcal{C}}}\|} + H\frac{\vect{T}_{\vect{\mathcal{C}}}}{\|\vect{T}_{\vect{\mathcal{C}}}\|}, \quad
        \boldsymbol{\Phi}_T = -\Pi_T\frac{\partial \boldsymbol{D}_{\vect{\mathcal{C}}}}{\partial t},
    \label{eq:static_and_time_field}
\end{align}
where $\Pi_T$ is the null space projector of $\vect{T}_{\vect{\mathcal{C}}}$.

To weight the sum of the components defined above, another scalar gain $\eta(\bvect{p},t) \in [0, v_r]$ is defined to ensure a desired reference speed $\|\vect{\Phi}(\bvect{p},t)\| = v_r$. This gain is given by
\begin{gather}
    \eta = -\vect{\Phi}_S \cdot \vect{\Phi}_T + \sqrt{(\vect{\Phi}_S \cdot \vect{\Phi}_T)^2 + v_r^2 - \|\vect{\Phi}_T\|^2}.
\end{gather}
From \cite{rezende2021constructive}, $\eta$ is defined for curves that move slower than the system velocity, $\|\boldsymbol{\Phi}_T\|<v_r$.

Finally, the field $\vect{\Phi}(\bvect{p},t)$ is given by
\begin{equation}
    \begin{array}{rcl}
        \vect{\Phi}(\bvect{p},t) & = & \eta(\bvect{p},t)\vect{\Phi}_S(\bvect{p},t) + \vect{\Phi}_T(\bvect{p},t).
    \end{array}
    \label{eq:field}
\end{equation}

According to Proposition 2 of \cite{rezende2021constructive}, let $V_p: \mathbb{R}^n \rightarrow \mathbb{R}$ be given by 
    \begin{equation}
        V_p(\boldsymbol{D}_{\boldsymbol{\mathcal{C}}}) = \frac{1}{2} \boldsymbol{D}_{\boldsymbol{\mathcal{C}}}\trans \boldsymbol{D}_{\boldsymbol{\mathcal{C}}},\label{eq:V_p_defintion}
    \end{equation}
    with $\boldsymbol{D}_{\boldsymbol{\mathcal{C}}}$ from \eqref{eq:curve_distance_vector}.
   When $\bvect{v} = \vect{\Phi}(\bvect{p},t)$ the following result holds
    \begin{align}
    \Dot{V}_{p}(\boldsymbol{D}_{\boldsymbol{\mathcal{C}}}) & = -\eta G \|\boldsymbol{D}_{\boldsymbol{\mathcal{C}}}\|,\label{eq:Dot_V_p_result}
    \end{align}
Which implies that $\boldsymbol{D}_{\boldsymbol{\mathcal{C}}} \rightarrow 0$ so trajectories of the system will converge to the curve. As $\boldsymbol{D}_{\boldsymbol{\mathcal{C}}} \rightarrow 0$, the circulation component in \eqref{eq:static_and_time_field} becomes dominant and the curve will be followed.

\subsubsection{Step 2}
For the second step, let $\vect{z}_v \equiv \Bar{\vect{v}} - \vect{\Phi}$ be a velocity error. The augmented system is then
\begin{align}
    \Dot{\vect{D}}_{\vect{\mathcal{C}}} & = \vect{\Phi} - \Dot{\vect{c}}_* + \vect{z}_v\\
    \Dot{\vect{z}}_v & = \Bar{\vect{a}} - \Dot{\vect{\Phi}}.
\end{align}

To control this system to follow the path, the following law is proposed:
\begin{align}
    \Bar{\vect{a}}_{\vect{\Phi}} & = - k_p\vect{D}_{\vect{\mathcal{C}}} - k_v\vect{z}_v + \Dot{\vect{\Phi}}, \quad k_p,k_v > 0.
    \label{eq:control_law_acceleration_curve_following}
\end{align}
\begin{lemma}[Nominal path-following]
    For the nominal system \eqref{eq:system_double_integrator_nominal}, if applied the control law $\Bar{\vect{a}} = \Bar{\vect{a}}_{\vect{\Phi}}$ from \eqref{eq:control_law_acceleration_curve_following} the trajectories of the system converge to and follow the target curve $\vect{\mathcal{C}}(t)$.
    \label{lemma:nominal_path_following}
\end{lemma}
\begin{proof}
    Let $V_{p,v}: \mathbb{R}^n \times \mathbb{R}^n \rightarrow \mathbb{R}$ be a candidate Lyapunov function given by
    \begin{align}
    V_{p,v}(\vect{D}_{\vect{\mathcal{C}}},\vect{z}_v) & = k_pV_p + \frac{1}{2}\vect{z}_v\trans\vect{z}_v, \quad k_p > 0,
\end{align}
with $V_p$ from \eqref{eq:V_p_defintion}.

Its time derivative is
\begin{align}
    \Dot{V}_{p,v} & = -k_p\eta G \|\vect{D}_{\vect{\mathcal{C}}}\| + k_p\vect{D}_{\vect{\mathcal{C}}}\trans\vect{z}_v + \vect{z}_v\trans\left(\Bar{\vect{a}} - \Dot{\vect{\Phi}}\right).
\end{align}

Applying the nominal path-following control law with  $\Bar{\vect{a}} = \Bar{\vect{a}}_{\vect{\Phi}}$ \eqref{eq:control_law_acceleration_curve_following} yields
\begin{align}
    \Dot{V}_{p,v} & = -k_p\eta G \|\vect{D}_{\vect{\mathcal{C}}}\| - k_v\|\vect{z}_v\|^2.
\end{align}
Which implies that $\boldsymbol{D}_{\boldsymbol{\mathcal{C}}} \rightarrow 0$ and $\vect{z}_v \rightarrow 0$ so trajectories of the system will converge to and follow the curve.
\end{proof}
\subsection{Collision avoidance}

For the next part of our method, we complete the nominal control strategy
by incorporating a nominal obstacle avoidance for
the nominal system \eqref{eq:system_double_integrator_nominal},

Let $\Bar{\lambda}$ be a safety threshold such that $\Bar{\lambda} > \lambda$.
The idea in this section is to ensure that
$\bvect{p}$ is at least $\Bar{\lambda}$ distant from $\vect{\mathcal{O}}$.
Later in the robust control formulation, we will show that this implies
that $\vect{p}$ will be at least $\lambda$ distant from the obstacles,
solving Problem \ref{problem:collision_avoidance}.

In order to propose the safety control law, we use a smooth half-squared distance function
$D_{\boldsymbol{\mathcal{O}}}^{h}(\boldsymbol{p},t)$, depending on a smoothing parameter $h > 0$,
given by \cite{nunes2026safe}
\begin{align}
    D_{\boldsymbol{\mathcal{O}}}^{h}(\boldsymbol{p},t)& \equiv \left(\sum_{i=1}^N
     \left(\frac{\|\boldsymbol{p} - \boldsymbol{o}_i(t)\|^2}{2}\right)^{-\frac{1}{h}}
     \right)^{-h},
    \label{eq:smooth_distance_definition}
\end{align}
with the property that, as shown in \cite{nunes2026safe},
$D_{\boldsymbol{\mathcal{O}}}^{h}(\boldsymbol{p},t) \le D_{\boldsymbol{\mathcal{O}}}(\boldsymbol{p},t)$.

Let
\begin{equation}
    B^{h,\Bar{\lambda}}(\bvect{p},t) = D_{\vect{\mathcal{O}}}^h(\bvect{p},t) - \frac{\Bar{\lambda}^2}{2},
    \label{eq:smooth_barrier}
\end{equation}
similarly to $B^{\lambda}(\vect{p},t)$ in \eqref{eq:safety_set_barrier_function},
but now using the smooth distance function \eqref{eq:smooth_distance_definition}.
Our goal for the collision avoidance control strategy is to
determine the acceleration $\Bar{\vect{a}}$ in \eqref{eq:system_double_integrator_nominal}
that will render $B^{h,\Bar{\lambda}}(\bvect{p},t) \geq 0$, $\forall t \geq 0$.
This can be achieved by imposing
the following differential inequality:
\begin{equation}
    \Ddot{B}^{h,\Bar{\lambda}} + k_1\Dot{B}^{h,\Bar{\lambda}} + k_2{B}^{h,\Bar{\lambda}} \ge 0,\label{eq:keep_distance_inequality}
\end{equation}
with $k_2 > 0$, $k_1 > 2\sqrt{k_2}$; and
initial conditions $B^{h,\Bar{\lambda}}(\bvect{p},0) \geq 0$,
and $\Dot{B}^{h,\Bar{\lambda}}(\bvect{p},\bvect{v},0) \geq r_1 B^{h,\Bar{\lambda}}(\bvect{p},0)$,
where $r_1$ is the smallest negative real root of
the Hurwitz polynomial $r^2 + k_1 r + k_2 = 0$.

From \eqref{eq:smooth_barrier}, by defining
the vector 
\begin{equation}
\vect{D}_{\vect{\mathcal{O}}}^h(\Bar{\vect{p}},t) \equiv \dfrac{\partial}{\partial \Bar{\vect{p}}} D_{\vect{\mathcal{O}}}^h(\Bar{\vect{p}},t)
\end{equation}
that points towards the direction of maximum increase in the distance from the obstacles, and
$E^h(\Bar{\vect{p}},\Bar{\vect{v}},\Bar{\vect{a}},t) = \Ddot{B}^{h,\Bar{\lambda}} + k_1\Dot{B}^{h,\Bar{\lambda}} + k_2{B}^{h,\Bar{\lambda}}$,
inequality \eqref{eq:keep_distance_inequality} can be rewritten as
\begin{align}
E^h(\Bar{\vect{p}},\Bar{\vect{v}},\Bar{\vect{a}},t) & = \vect{D}_{\vect{\mathcal{O}}}^h(\Bar{\vect{p}},t)\trans\Bar{\vect{a}} + F_1(\Bar{\vect{p}},\Bar{\vect{v}},t) \geq 0, \label{eq:E_h_defintion}\\
F_1(\Bar{\vect{p}},\Bar{\vect{v}},t) & = \Bar{\vect{v}}\trans \left(\dfrac{\partial}{\partial\Bar{\vect{p}}}\vect{D}_{\vect{\mathcal{O}}}^h(\Bar{\vect{p}},t)\right)\Bar{\vect{v}}\nonumber\\&  + 2 \dfrac{\partial}{\partial t}\vect{D}_{\vect{\mathcal{O}}}^h(\Bar{\vect{p}},t)\trans \Bar{\vect{v}} + \dfrac{\partial^2}{\partial t^2}D_{\vect{\mathcal{O}}}^h(\Bar{\vect{p}},t) \nonumber\\
& + k_1\Dot{B}^{h,\Bar{\lambda}}(\Bar{\vect{p}},\Bar{\vect{v}},t) + k_2B^{h,\Bar{\lambda}}(\Bar{\vect{p}},t).\nonumber
\end{align}

In addition to not colliding, we also want the system to circulate obstacles.
To achieve this, let $\vect{T}_{\vect{\mathcal{O}}}(\Bar{\vect{p}},t) = M\vect{D}_{\vect{\mathcal{O}}}^h(\Bar{\vect{p}},t)$
be a vector, where $M$ is a pre-defined skew-symetric matrix \cite{nunes2026safe},
such that $\vect{D}_{\vect{\mathcal{O}}}^h(\Bar{\vect{p}},t) \perp \vect{T}_{\vect{\mathcal{O}}}(\Bar{\vect{p}},t)$
and $C(\Bar{\vect{p}},\Bar{\vect{v}},t) \equiv \vect{T}_{\vect{\mathcal{O}}}(\Bar{\vect{p}},t)\trans\Bar{\vect{v}} > 0$
(the current velocity has a positive component in the direction of $\vect{T}_{\vect{\mathcal{O}}}(\Bar{\vect{p}},t)$).
To ensure the circulation condition it is sufficient to have
\begin{equation}
    \Dot{C} > 0, \label{eq:circulation_inequality}
\end{equation}
where,
\begin{align}
    \label{eq:circulation}
    \Dot{C}(\Bar{\vect{p}},\Bar{\vect{v}},\Bar{\vect{a}},t) & = \vect{T}_{\vect{\mathcal{O}}}(\Bar{\vect{p}},t)\trans\Bar{\vect{a}} + F_2(\Bar{\vect{p}},\Bar{\vect{v}},t) > 0,\\
    F_2(\Bar{\vect{p}},\Bar{\vect{v}},t) & =  \Bar{\vect{v}}\trans \left(\dfrac{\partial}{\partial\Bar{\vect{p}}}\vect{T}_{\vect{\mathcal{O}}}(\Bar{\vect{p}},t)\right)\Bar{\vect{v}} + \dfrac{\partial}{\partial t}\vect{T}_{\vect{\mathcal{O}}}(\Bar{\vect{p}},t)\trans \Bar{\vect{v}}.\nonumber
\end{align}
From the previous development,
finding the collision avoidance strategy can be framed as the search for a target acceleration
$\Bar{\vect{a}}_{\vect{\Psi}}$, such that
$\Bar{\vect{a}}_{\vect{\Psi}}  = \arg \min_{\Bar{\vect{a}}} \|\Bar{\vect{a}}\|^2$,
subjected to conditions \eqref{eq:E_h_defintion} and \eqref{eq:circulation}, respectively.
The solution of this optimization problem is then
\begin{equation}
    \Bar{\vect{a}}_{\vect{\Psi}} = \frac{-F_1}{\left\| \vect{D}_{\vect{\mathcal{O}}}^h\right\|^2}
\vect{D}_{\vect{\mathcal{O}}}^h + \left(\frac{-F_2}{\left\| \vect{T}_{\vect{\mathcal{O}}}\right\|^2} + \varepsilon\right)
\vect{T}_{\vect{\mathcal{O}}}, \quad \varepsilon > 0,
    \label{eq:control_law_acceleration_deviate_obstacle}
\end{equation}
such that
\begin{equation}
 E^h(\Bar{\vect{p}},\Bar{\vect{v}},\Bar{\vect{a}}_{\vect{\Psi}},t) = 0, \quad \Dot{C}(\Bar{\vect{p}},\Bar{\vect{v}},\Bar{\vect{a}}_{\vect{\Psi}},t) > 0.\label{eq:E_h_only_avoidance_equals_0}
\end{equation}


\subsection{Combined strategy}

Considering $\Bar{\vect{a}}_{\vect{\Phi}}$ in \eqref{eq:control_law_acceleration_curve_following}
and $\Bar{\vect{a}}_{\vect{\Psi}}$ in \eqref{eq:control_law_acceleration_deviate_obstacle},
let $\Bar{\vect{a}}$ be a blended control law
\begin{align}
\Bar{\vect{a}} &  = \Theta\Bar{\vect{a}}_{\vect{\Phi}} + (1-\Theta)\Bar{\vect{a}}_{\vect{\Psi}},
\label{eq:a_control_law_nominal}\\
\Theta &= 
\mathrm{sat}_0^1\left(k_E\vect{D}_{\vect{\mathcal{O}}}^h(\Bar{\vect{p}},t)\trans\left(\Bar{\vect{a}}_{\vect{\Phi}} - \Bar{\vect{a}}_{\vect{\Psi}}\right)\right),
\label{eq:theta_original}
\end{align}
with $k_E > 0$ and $\mathrm{sat}_a^b(x) = \min(b,\max(a,x))$ is a saturation function.

\begin{lemma}(Nominal safety)
For the nominal system \eqref{eq:system_double_integrator_nominal}, the blended control law $\bvect{a}$ in \eqref{eq:a_control_law_nominal} keeps the system at least a $\Bar{\lambda}$ distance from the obstacles.
\label{lemma:nominal_safety}
\end{lemma}
\begin{proof}
Applying the control law $\Bar{\vect{a}}$ \eqref{eq:a_control_law_nominal} in the barrier certificate $E^h$ \eqref{eq:E_h_defintion} yields
\begin{align}
E^h(\Bar{\vect{p}},\Bar{\vect{v}},\Bar{\vect{a}},t) & = \Theta (\vect{D}_{\vect{\mathcal{O}}}^h(\Bar{\vect{p}},t)\trans \Bar{\vect{a}}_{\vect{\Phi}}) -\Theta(\vect{D}_{\vect{\mathcal{O}}}^h(\Bar{\vect{p}},t)\trans\Bar{\vect{a}}_{\vect{\Psi}}) \nonumber\\
& + \vect{D}_{\vect{\mathcal{O}}}^h(\Bar{\vect{p}},t)\trans \Bar{\vect{a}}_{\vect{\Psi}} + F_1(\Bar{\vect{p}},\Bar{\vect{v}},t).\label{eq:E_h_expanded_control_law}
\end{align}

Inspecting \eqref{eq:theta_original}, we conclude that if $\Theta = 0$, then $\Bar{\vect{a}} = \Bar{\vect{a}}_{\vect{\Psi}}$ and $E^h(\Bar{\vect{p}},\Bar{\vect{v}},\Bar{\vect{a}},t) = 0$
as ensured by the result \eqref{eq:E_h_only_avoidance_equals_0}. If $\Theta = 1$, then $\Bar{\vect{a}} = \Bar{\vect{a}}_{\vect{\Phi}}$ and $\vect{D}_{\vect{\mathcal{O}}}^h(\Bar{\vect{p}},t)\trans\left(\Bar{\vect{a}}_{\vect{\Phi}} - \Bar{\vect{a}}_{\vect{\Psi}}\right) = E^h(\Bar{\vect{p}},\Bar{\vect{v}},\Bar{\vect{a}},t) \ge \frac{1}{k_E}$ as ensured by the value of $\Theta$ itself. For values of $\Theta$ between zero and one, we must analyze further.

The term $\vect{D}_{\vect{\mathcal{O}}}^h(\Bar{\vect{p}},t)\trans\Bar{\vect{a}}_{\vect{\Psi}} + F_1(\Bar{\vect{p}},\Bar{\vect{v}},t)$ is null according to \eqref{eq:E_h_only_avoidance_equals_0}. 
Therefore
%
\begin{equation}
    E^h = k_E\left(\vect{D}_{\vect{\mathcal{O}}}^h(\Bar{\vect{p}},t)\trans \Bar{\vect{a}}_{\vect{\Phi}} - \vect{D}_{\vect{\mathcal{O}}}^h(\Bar{\vect{p}},t)\trans\Bar{\vect{a}}_{\vect{\Psi}}\right)^2 \ge 0.
\end{equation}
Since $E^h \ge 0$, the system stays at the set defined by $B^{h,\Bar{\lambda}}(\Bar{\vect{p}},t)\ge 0$, according to the barrier function framework \cite{8796030}.

According to Lemma 2 of \cite{nunes2026safe}, if the smoothed barrier function holds $B^{h,\Bar{\lambda}}(\Bar{\vect{p}},t)\ge 0$, then the standard barrier function also holds because $B^{\Bar{\lambda}}(\Bar{\vect{p}},t)\ge B^{h,\Bar{\lambda}}(\Bar{\vect{p}},t) \ge 0$, since $D_{\vect{\mathcal{O}}}(\bvect{p},t) \geq D_{\vect{\mathcal{O}}}^{h}(\bvect{p},t)$. This implies that the nominal system will keep at least a safety distance $\Bar{\lambda}$ from the obstacle set.
\end{proof}

\section{Robust adaptive control}

For the last part of our methodology, we propose a robust adaptive Backstepping control law to stabilize the error around zero between the practical system \eqref{eq:system_double_integrator} and the nominal system \eqref{eq:system_double_integrator_nominal}.

The \textbf{error system} is a representation of the difference between those systems.

Let $\vect{e}_p = \vect{p} - \Bar{\vect{p}}$ be the error between the practical and nominal positions. Let $\vect{v}_e = \vect{v} - \Bar{\vect{v}}$ be an auxiliary variable that represents the error between the practical and nominal velocities such that
\begin{align}
    \Dot{\vect{e}}_p & = \vect{v}_e.
    \label{eq:p_error_system_first_order}
\end{align}

\subsection{Step 1}

We want $\vect{v}_e$ to follow a virtual control law $\vect{\mathcal{F}}_e$ such that $\vect{e}_p \rightarrow 0$ if $\vect{v}_e = \vect{\mathcal{F}}_e$. This virtual control law is chosen as 
\begin{align}
    \vect{\mathcal{F}}_e = -k_{\mathcal{F}e}\vect{e}_p, \quad k_{\mathcal{F}e} > 0.
    \label{eq:virtual_control_law_F_e}
\end{align}

Applying the virtual control law \eqref{eq:virtual_control_law_F_e} in the system \eqref{eq:p_error_system_first_order}, the point $\vect{e}_p = 0$ becomes an asymptotic stable equilibrium point as $\Dot{\vect{e}}_p =  -k_{\mathcal{F}e}\vect{e}_p$.



\subsection{Step 2}

Let $\vect{e}_v = \vect{v}_e - \vect{\mathcal{F}}_e$ be another velocity error that encompasses the virtual control law $\vect{\mathcal{F}}_e$.
The error system is represented by
\begin{align}
    \Dot{\vect{e}}_p & = \vect{\mathcal{F}}_e + \vect{e}_v,\label{eq:pos_error_system1}\\
    \Dot{\vect{e}}_v & = \vect{a}_e + \vect{a}_d - \Dot{\vect{\mathcal{F}}}_e,
    \label{eq:pos_error_system2}
\end{align}
where $\vect{a}_e \equiv \vect{a} - \Bar{\vect{a}}$ is its control input.

In order to state the robust adaptive control law, consider the following assumption.

\begin{assumption}[Norm bounded disturbances]
Let
   \begin{align}
    \chi_p & =  \vect{a}_d - \Dot{\vect{\mathcal{F}}}_e
    \label{eq:lumped_term_position}
\end{align}
be a lumped term that contains the disturbance \cite{Zhang2021}. Using an indirect method \cite{5699388,Cai2008}, the controller tackles $\chi_p$ by focusing only on its bound. Therefore, we assume that exists unknown constants $\zeta_1$ and $\zeta_2$ such that 
\begin{equation}
    \|\chi_p\| \leq \zeta_1 + \zeta_2\|\vect{e}_v\|,
\end{equation} 
and those unknown constants $\zeta \equiv \begin{bmatrix} \zeta_1 & \zeta_2\end{bmatrix}\trans$ belong to a convex set $\Xi$, \emph{i.e.}
\begin{equation}
   \zeta \in \Xi, \quad \Xi = \{\xi \in \mathbb{R}^2 \ | \ \alpha(\xi) \leq 0\},
\end{equation}
where $\alpha: \mathbb{R}^2 \rightarrow \mathbb{R}$ is a smooth convex function.
\label{assumption:norm_bounded_disturnaces_p}
\end{assumption}

Those constants that bound $\chi_p$ are estimated in the controller by adaptive gains $\Hat{\zeta}_1$ and $\Hat{\zeta}_2$ \cite{Zhang2021}. 

Let $\vect{F}: \mathbb{R}^n \rightarrow \mathbb{R}^n$ be the element-wise version of a function $F: \mathbb{R} \rightarrow \mathbb{R}$ such that: $\lim_{x\rightarrow\infty}F(x) = 1$; $\lim_{x\rightarrow-\infty}F(x) = -1$; and $F(0) = 0$.

Then, the robust adaptive control law for the error system is given by
\begin{align}
    \vect{a}_e = -k_{ep}\vect{e}_p - k_{ev}\vect{e}_v - \Hat{\zeta}_1\vect{F}(\Hat{\zeta}_1\vect{e}_v/\epsilon_p) - \Hat{\zeta}_2\vect{e}_v,
    \label{eq:control_law_f_e}
\end{align}
where $k_{ep},k_{ev},\epsilon_p>0$ are scalar gains. For notation simplification, assume that $\alpha \equiv \alpha(\Hat{\zeta})$. The adaptive parameters have the following update laws, where a projection modification was applied \cite{farrell2006adaptive},
\begin{align}
    \Dot{\Hat{\zeta}} & = \begin{cases} P\epsilon  & \text{if } \alpha < 0 \text{ or } \nabla\alpha\trans P\epsilon\leq 0 ,\\
    P\epsilon - P\frac{\nabla\alpha\nabla\alpha\trans }{\nabla\alpha\trans P\nabla\alpha}P\epsilon  & \text{otherwise}
    \end{cases}\nonumber
    \\
    P & = \begin{bmatrix}
        \rho_1 & 0 \\ 0 & \rho_2
    \end{bmatrix} > 0, \quad \epsilon = \begin{bmatrix}
        \|\boldsymbol{e}_v\| \\ \|\boldsymbol{e}_v\|^2
    \end{bmatrix}.\label{eq:pos_zetas_law}
\end{align}

\begin{lemma}
    Considering Assumption \ref{assumption:norm_bounded_disturnaces_p}, let $V_{epv}$ be the following control Lyapunov function
    \begin{align}
        V_{epv} & = k_{ep}V_{ep} + \dfrac{1}{2}\vect{e}_v \cdot \vect{e}_v + \dfrac{\Tilde{\zeta}_1^2}{2\rho_1} + \dfrac{\Tilde{\zeta}_2^2}{2\rho_2}, \label{eq:Vepv}
    \end{align}
    where 
    \begin{equation}
             \Tilde{\zeta} \equiv \begin{bmatrix}
                 \Tilde{\zeta}_1 & \Tilde{\zeta}_2
             \end{bmatrix}\trans  = \zeta - \Hat{\zeta}
    \end{equation}
    are estimation errors. By applying \eqref{eq:control_law_f_e} in the system \eqref{eq:pos_error_system1}-\eqref{eq:pos_error_system2}, $\Dot{V}_{epv} < 0$ if the system is outside the attractive set $\Delta_p$ given by
    \begin{align}
        \Delta_p & = \{(\vect{e}_p,\vect{e}_v) \ | \ \|\vect{e}_p\| \leq \delta_p, \|\vect{e}_v\| \leq \delta_v \},\nonumber\\
        \delta_p & = \sqrt{\dfrac{n\kappa\epsilon_p}{k_{ep}k_{\mathcal{F}e}}},\
        \delta_v = \sqrt{\dfrac{n\kappa\epsilon_p}{k_{ev}}},\label{eq:pos_tube}
    \end{align}
    with $n$ the problem dimension $\mathbb{R}^n$ and 
    \begin{align}
        \kappa & = \max_x |x| - xF(x), \quad x \in \mathbb{R}.
    \end{align}
    \label{lemma:pos_tube_size}
\end{lemma}
\begin{proof}
    The time derivative of $V_{epv}$ \eqref{eq:Vepv} is given by
\begin{align}
     \Dot{V}_{epv} & = -k_{ep}k_{\mathcal{F}e} \|\vect{e}_p\|^2 + k_{ep}\vect{e}_p\cdot\vect{e}_v  \\
     & + \vect{e}_v\cdot(\vect{a}_d - \Dot{\vect{\mathcal{F}}}_e + \vect{a}_e) - \Tilde{\zeta_1}\Dot{\Hat{\zeta}}_1/\rho_1  - \Tilde{\zeta_2}\Dot{\Hat{\zeta}}_2/\rho_2.
    \label{eq:dot_Vepv}
\end{align}

For the sake of simplicity, consider for now that the first conditions for $\Dot{\Hat{\zeta}}$ law in \eqref{eq:pos_zetas_law} are always met. Later we show that the projection modification \cite{farrell2006adaptive} does not impact on the invariant set for the error system. Then, substituting \eqref{eq:control_law_f_e} and \eqref{eq:pos_zetas_law} in \eqref{eq:dot_Vepv} yields
\begin{align}     
    \Dot{V}_{epv} & = -k_{ep}k_{\mathcal{F}e} \|\boldsymbol{e}_p\|^2 -k_{ev} \|\boldsymbol{e}_v\|^2 + \boldsymbol{e}_v \cdot \chi_p\nonumber\\
    & - \boldsymbol{e}_v\cdot\Hat{\zeta}_1\vect{F}(\Hat{\zeta}_1\boldsymbol{e}_v/\varepsilon_p) - \Hat{\zeta}_2\|\boldsymbol{e}_v\|^2
    - \Tilde{\zeta_1}\|\boldsymbol{e}_v\| - \Tilde{\zeta_2}\|\boldsymbol{e}_v\|^2.
\nonumber
\end{align}
Considering Assumption \ref{assumption:norm_bounded_disturnaces_p} leads to
\begin{align}
    \Dot{V}_{epv} & \leq -k_{ep}k_{\mathcal{F}e} \|\boldsymbol{e}_p\|^2 -k_{ev} \|\boldsymbol{e}_v\|^2\nonumber\\
    &- \Hat{\zeta}_1\|\boldsymbol{e}_v\|\vect{F}(\Hat{\zeta}_1\|\boldsymbol{e}_v\|\varepsilon_p) + \Hat{\zeta}_1\|\boldsymbol{e}_v\|.\nonumber 
\end{align}
Using the following properties
\begin{align}
    \|\boldsymbol{q}\| & \leq \sum_{i=1}^n|q_i|, \forall \boldsymbol{q} \in \mathbb{R}^n,\label{eq:property_vector_norm}\\
    0 & \leq |q| - qF(q/\varepsilon_p) \leq \kappa\varepsilon_p, \forall q \in \mathbb{R},\label{eq:property_tanh_bound}
\end{align}
yields:
\begin{align}
    \Dot{V}_{epv} \leq -k_{ep}k_{\mathcal{F}e}\|\boldsymbol{e}_p\|^2 - k_{ev}\|\boldsymbol{e}_v\|^2 + n\kappa \varepsilon_p.
    \label{eq:lyapunov_robust_upper_bound}
\end{align}
in which $\Dot{V}_{epv} < 0$ for $\|\boldsymbol{e}_p\| > \sqrt{\frac{n\kappa\varepsilon_p}{k_{ep}k_{\mathcal{F}e}}}$ or $\|\boldsymbol{e}_v\| > \sqrt{\frac{n\kappa\varepsilon_p}{k_{ev}}}$.

When the projection modification is active, \emph{i.e.} the second conditions for $\Dot{\Hat{\zeta}}$ law in \eqref{eq:pos_zetas_law} are met, there is an extra term in the Lyapunov equation derivative $\Dot{V}_{epv}$ given by
\begin{align}
    \Tilde{\zeta}\trans \frac{\nabla\alpha\nabla\alpha\trans }{\nabla\alpha\trans P\nabla\alpha}P\epsilon.
    \label{eq:extra_term_projection_modification}
\end{align}
By definition, $\nabla\alpha(\Hat{\zeta})\trans P\epsilon > 0$ and $\Tilde{\zeta}\trans \nabla\alpha(\Hat{\zeta}) = (\zeta - \Hat{\zeta})\trans \nabla\alpha(\Hat{\zeta}) \leq \alpha(\zeta) - \alpha(\Hat{\zeta})\leq0$ since $\Xi$ is convex and $\zeta \in \Xi$ \cite{farrell2006adaptive}. Therefore, the extra term \eqref{eq:extra_term_projection_modification} does not make \eqref{eq:lyapunov_robust_upper_bound} more positive and ensures that the adaptive gains remain in $\Xi$ \cite{farrell2006adaptive}. 
\end{proof}

The former inequality means that $\delta_p = \sqrt{\frac{n\kappa\epsilon_p}{k_{ep}k_{\mathcal{F}e}}}$ bounds an attractive set for the position error $\vect{e}_p$.

From the latter inequality and using the triangle inequality, $\|\vect{e}_v\| \geq \|\vect{v}_e\| - \|\vect{\mathcal{F}}_e\|$, with $\|\vect{\mathcal{F}}_e\| = k_{\mathcal{F}e}\|\vect{e}_p\| \leq k_{\mathcal{F}e}\delta_p$ and finally $\delta_{ve} = \delta_v + k_{\mathcal{F}e}\delta_p$ bounds an attractive set for the velocity error $\vect{v}_e$, with $\delta_p$ and $\delta_v$ from \eqref{eq:pos_tube}.

Now we are able to present the complete controller scheme, which is given by
\begin{equation}
  \vect{a} = \Bar{\vect{a}}+\vect{a}_e,
  \label{eq:f_complete}
\end{equation}
with $\Bar{\vect{a}}$ from \eqref{eq:a_control_law_nominal}, and with $\vect{a}_e$ from \eqref{eq:control_law_f_e}. The proposed method yields the following result:
\begin{theorem}[Controller]
    Applying \eqref{eq:f_complete} in the practical system \eqref{eq:system_double_integrator}, its trajectories: (i) respects the safety distance if $\Bar{\lambda} > \lambda + \delta_p$ and; (ii) when it is safe, asymptotically converge to a control attractive set around the target curve $\vect{\mathcal{C}}(t)$, in which the set is \eqref{eq:pos_tube}.
    \label{theorem:position_controller}
\end{theorem} 
\begin{proof}
    Lemma \ref{lemma:nominal_path_following} states that the nominal system \eqref{eq:system_double_integrator_nominal} asymptotically converges to the path when it is safe and Lemma \ref{lemma:nominal_safety} states the nominal system's distance to the obstacle set is kept greater than or equal to $\Bar{\lambda}$.
    
    Furthermore, the error between the practical system and the nominal trajectory is constrained in a small set as a result of Lemma \ref{lemma:pos_tube_size}. This means that the system converges to an attractive set around the path when it is safe and that it keeps a distance around $\Bar{\lambda}$ from the obstacles with a maximum deviation of $\delta_p$.
\end{proof}
Note that one can choose the value of $\lambda$ independently from $\delta_p$, such that it is always possible to make $\Bar{\lambda} > \lambda + \delta_p$ and therefore ensure safety.

\section{Results}

We validate our method through extensive simulations in different scenarios. A detailed visualization with animation videos of the results is available at \href{https://youtu.be/dkCGE_h0_EE}{youtu.be/dkCGE\_h0\_EE}. The simulation implementation is available at \href{https://github.com/ArthurHDN/position_controller}{github.com/ArthurHDN/position\_controller}. We divide the simulation batches in $\mathbb{R}^2$ and $\mathbb{R}^3$ scenarios for easier comparison. 

\subsection{$\mathbb{R}^2$ scenarios}

This set of simulations consisted of two different curves: (i) a static circle for baseline; and (ii) a smooth star curve that expands/contracts and moves its center. For each curve, four different sets of obstacles were used: (i) no obstacles; (ii) static sparse obstacles; (iii) static cluttered obstacles; (iv) moving obstacles. Figure \ref{fig:sim2dbatch_scenarios} shows the combinations of curves and obstacles used. Each row represents a curve and each column an obstacle set, where both the curve on the last row and the obstacle set on the last column were dynamic. 
\begin{figure}[htbp]
   \centering
    \includegraphics[clip,trim={0.0in 0.0in 0.0in 0.0in}, width=0.999\linewidth]{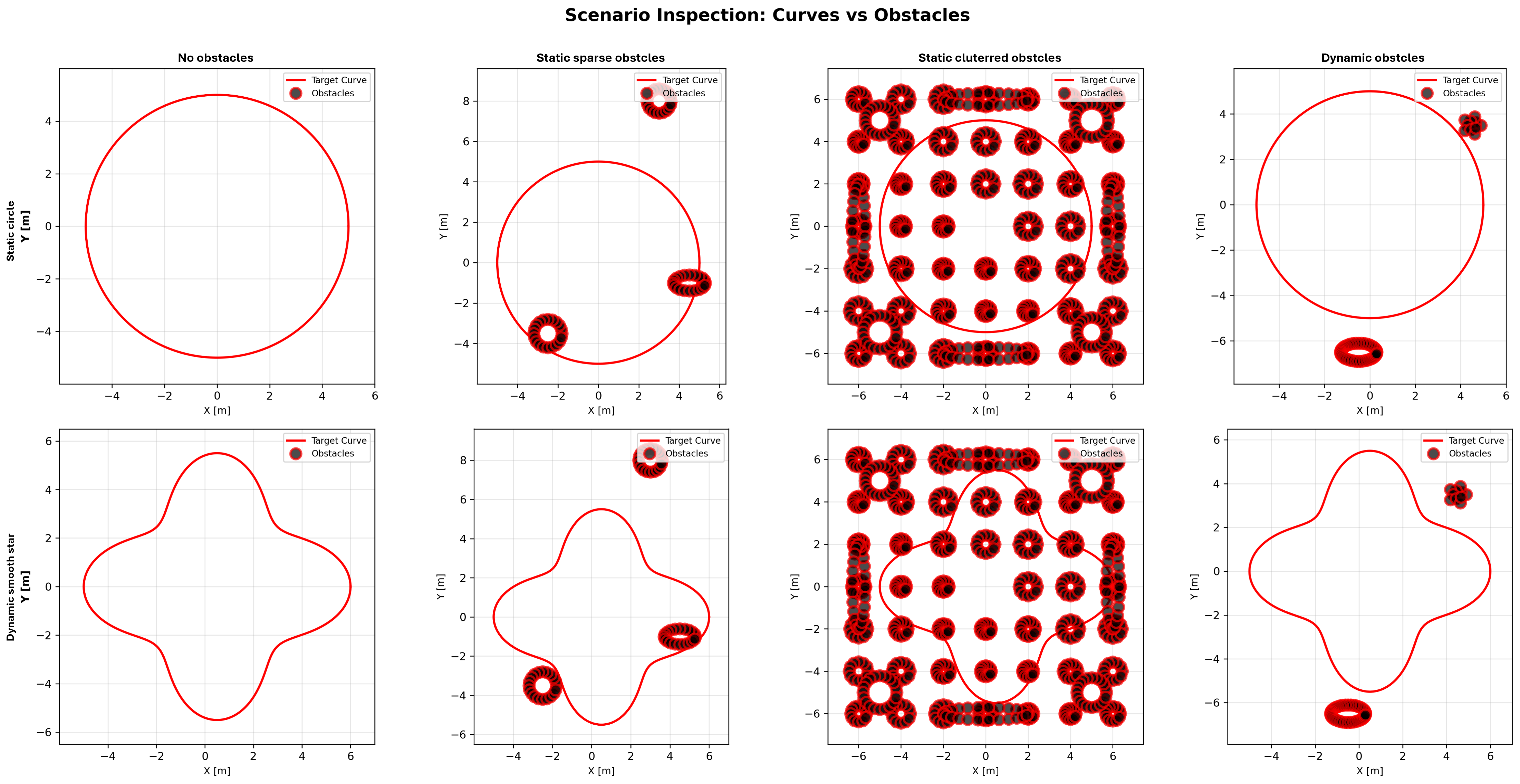}
   \caption{Combinations of curves and obstacles used for the 2D simulations.}
   \label{fig:sim2dbatch_scenarios}
\end{figure}

For each scenario in Fig. \ref{fig:sim2dbatch_scenarios}, three disturbance profiles were used: (i) null disturbances; (ii) moderate disturbances, $\begin{bmatrix}
  \mathcal{U}_{-0.1,0.1} & \mathcal{U}_{-0.1,0.1}   
\end{bmatrix}\trans + \begin{bmatrix}
    0.2\mathrm{s}_5 & 0.1\sin(0.1t)
\end{bmatrix}\trans$ with $\mathcal{U}_{a,b}$ being a uniformly distributed random sample from $a$ to $b$ and $\mathrm{s}_{t_1}$ being the unit step function for $t=t_1$; and (iii) high disturbance $\begin{bmatrix}
  \mathcal{U}_{-1,1} & \mathcal{U}_{-1,1}   
\end{bmatrix}\trans + \begin{bmatrix}
    3 & 1.5\sin(0.5t)
\end{bmatrix}\trans$. Figure \ref{fig:sim2dbatch_disturbances} shows one run for each of the disturbance profiles.
\begin{figure}[htbp]
   \centering
    \includegraphics[clip,trim={0.0in 0.0in 0.0in 0.0in}, width=0.99\linewidth]{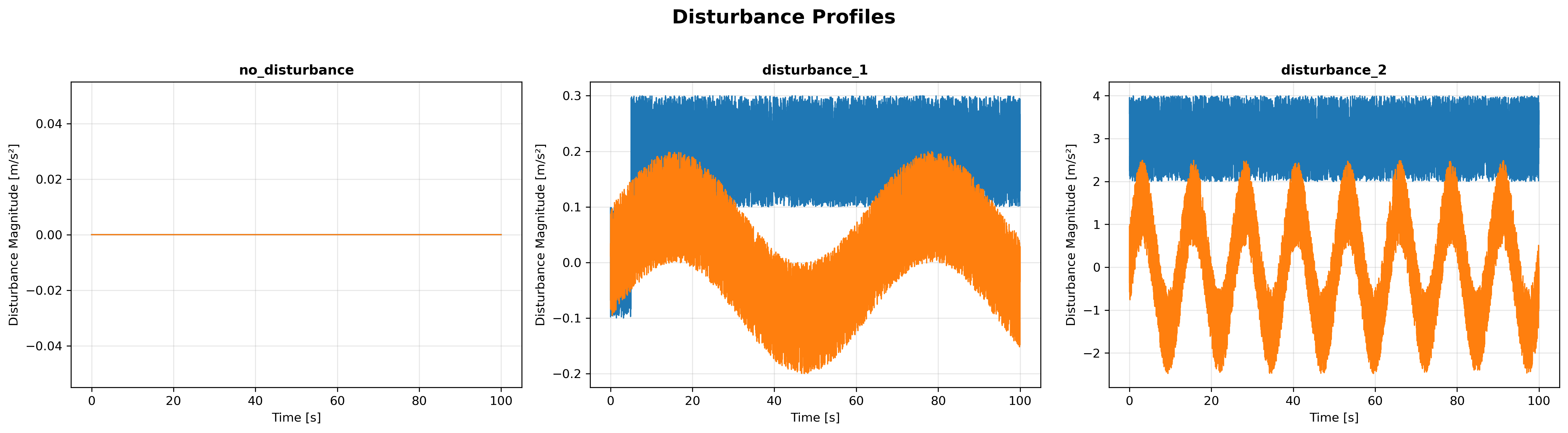}
   \caption{Disturbances profiles applied to the simulations.}
   \label{fig:sim2dbatch_disturbances}
\end{figure}

Finally, the nominal controller \eqref{eq:a_control_law_nominal} alone and the final controller \eqref{eq:f_complete} were applied in each curve-obstacle-disturbance scenario 50 times with random initial positions. This resulted in 2 curves, 4 obstacle profiles, 3 disturbance profiles, 2 controllers, and 50 runs, totalizing 2400 simulations. 

Figure \ref{fig:sim2dbatch_boxplot} shows boxplots of the results. In the top, the percentage of time where the system was inside the tube around the curve, \emph{i.e.} $\|\vect{D}_{\vect{\mathcal{C}}}(\vect{p},t)\| \leq \delta_p$. In the bottom, the percentage of the time where the safety constraint was violated, \emph{i.e.} $B^\lambda(\vect{p},t) < 0$. For both graphs, it separates the results for each obstacle profile in the $x$-axis and separates for each controller using different colors, red being the nominal controller only and blue the robust controller. 
\begin{figure}[htbp]
   \centering
    \includegraphics[clip,trim={0.0in 0.0in 0.0in 0.0in}, width=0.99\linewidth]{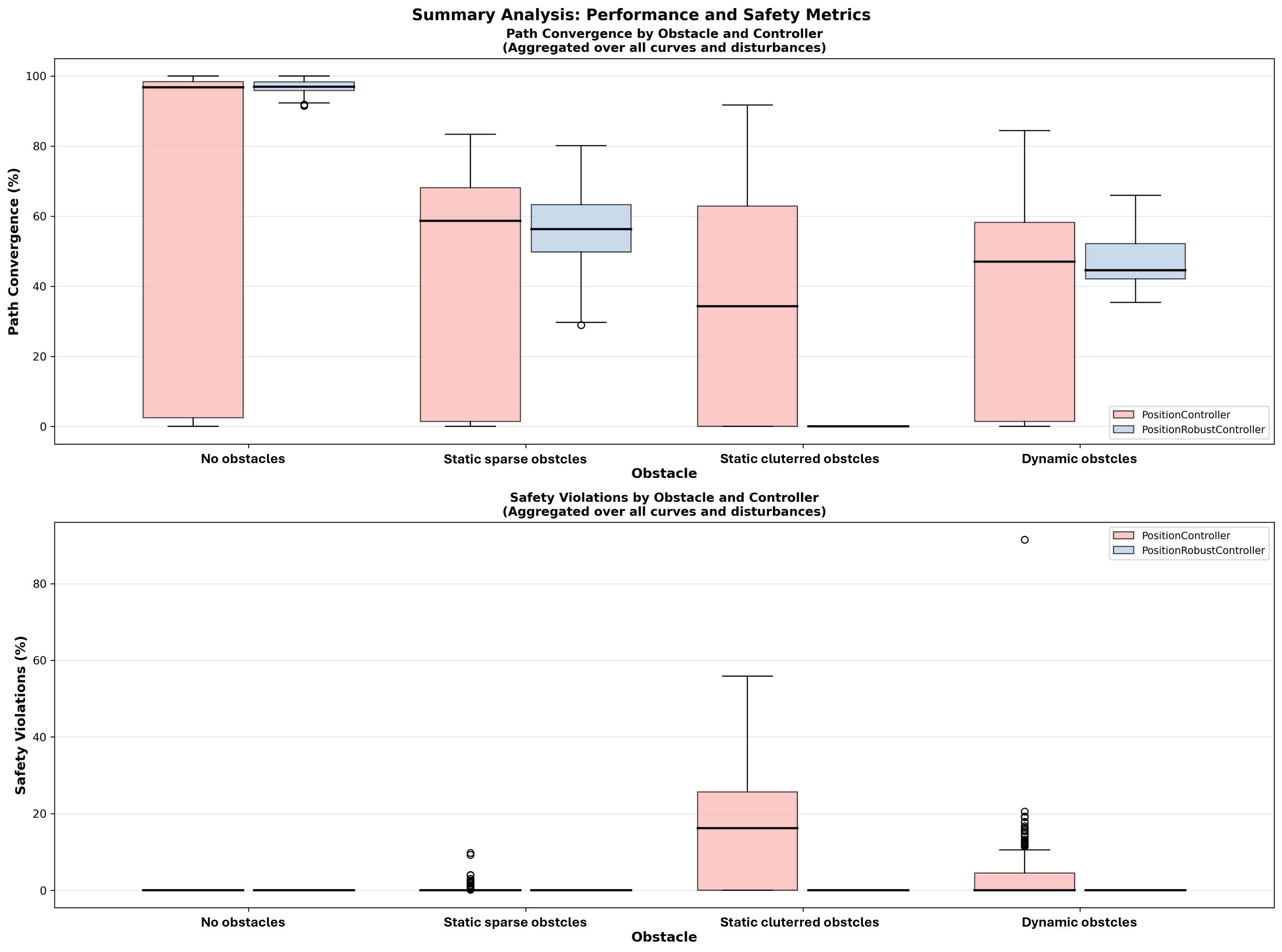}
    \caption{Boxplot of the simulations results showing the percentage of time where the system was inside the tube around the curve (top) and the percentage of the time where the safety constraint was violated (bottom). Results are separated for each obstacle profile and for each controller.}
   \label{fig:sim2dbatch_boxplot}
\end{figure}

Overall, the system behaved as expected for both static and dynamic cases. For the obstacle-free scenario, the mean of the path convergence time is nearly 100\%. In this case, the variation for the robust controller is small. This is notably different from the nominal controller, because in the high disturbed cases it could not stabilize the system and was dominated by the disturbance, which is also observed in the other obstacle cases. 

Notably, for the other obstacle profiles the system could not always follow the path, being forced to take detours in order to avoid collisions. The detours taken by the robust controller were higher than the nominal one, resulting in marginally less path convergence. Although possessing a marginally higher path convergence, the nominal controller did not respect the safety margin. This happened when the system was disturbed, what also lead to collisions on some simulations. The robust controller on the other hand always respected the safety margin. Therefore, the problems stated in this work were solved by the proposed controller.

\subsection{$\mathbb{R}^3$ scenarios}

This section shows another set of simulations following the same structure as previous, but this time for 3D problems. Figure \ref{fig:sim3dbatch_scenarios} shows the combinations of curves and obstacles used. Each row represents a curve, where the first is a static circle and the second is a moving trefoil knot. Each column represents an obstacle set, where the first set is no obstacle and the second is a moving cube. 
\begin{figure}[htbp]
   \centering
    \includegraphics[clip,trim={0.0in 0.0in 0.0in 0.0in}, width=0.99\linewidth]{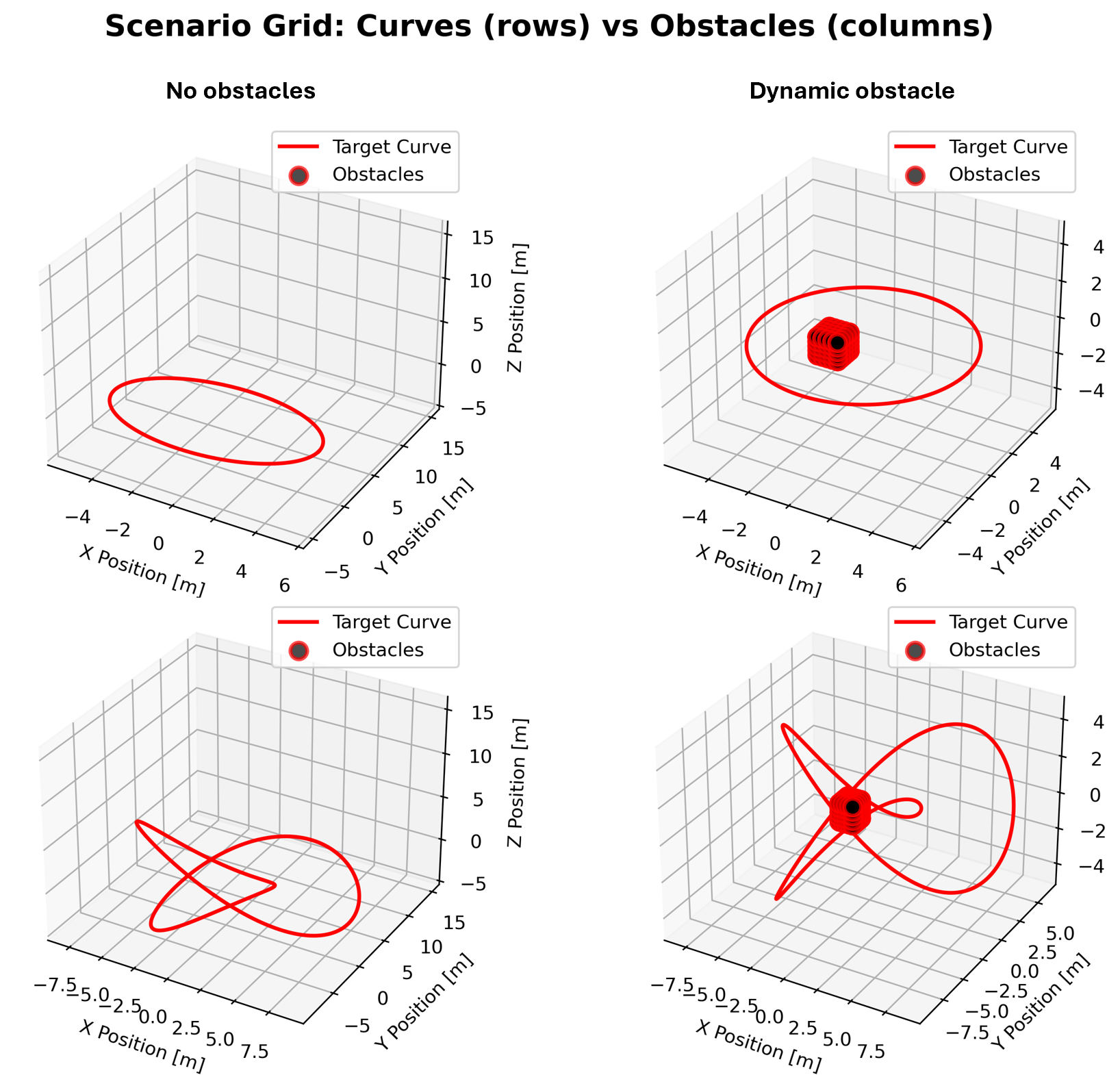}
   \caption{Combinations of curves and obstacles used for the 3D simulations.}
   \label{fig:sim3dbatch_scenarios}
\end{figure}

For each scenario in Fig. \ref{fig:sim3dbatch_scenarios}, two disturbance profiles were used: (i) null disturbances; and (ii) moderate disturbances, $$\begin{bmatrix}
  0.4\cos(0.1t) \\ 0.4\sin(0.1t) \\ 0.2\sin(0.15t)   
\end{bmatrix} + \begin{bmatrix}
    0.3\sin(2t)\cos(0.5t) \\ 0.3\cos(2.2t)\sin(0.6t) \\ 0.25\sin(1.8t)\cos(0.55t)
\end{bmatrix} + \begin{bmatrix}
    \mathcal{U}_{-0.15,0.15} \\ \mathcal{U}_{-0.15,0.15} \\ \mathcal{U}_{-0.15,0.15}
\end{bmatrix}.$$
Figure \ref{fig:sim3dbatch_disturbances} shows one run for each of the disturbance profiles.
\begin{figure}[htbp]
   \centering
    \includegraphics[clip,trim={0.0in 0.0in 0.0in 0.0in}, width=0.99\linewidth]{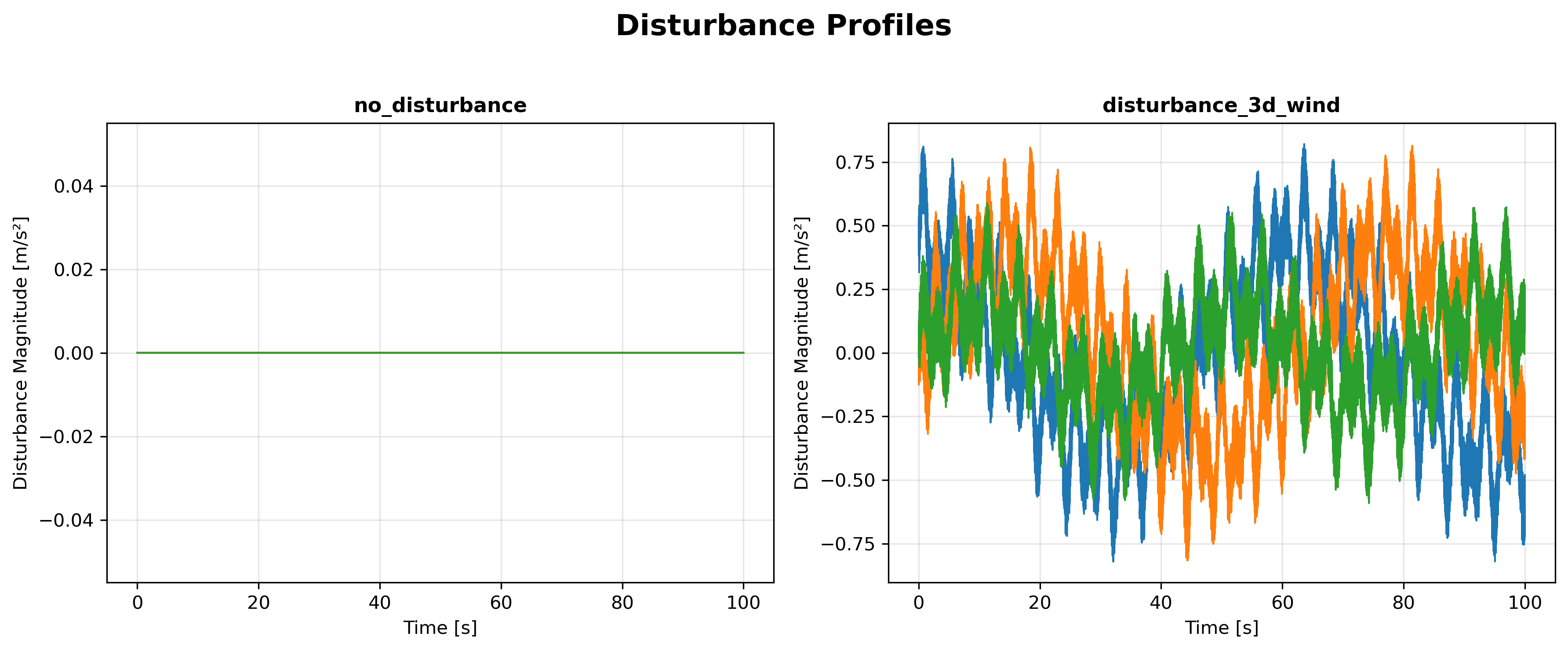}
   \caption{Disturbances profiles applied to the simulations.}
   \label{fig:sim3dbatch_disturbances}
\end{figure}

Overall the results mirror the previous ones, drawing the same conclusions. In the $\mathbb{R}^3$ scenarios, there were no safety violations. This is due the scenarios being mostly free with only one dynamic obstacle. Figure \ref{fig:sim3dbatch_boxplot} shows the boxplot analysis. 
\begin{figure}[htbp]
   \centering
    \includegraphics[clip,trim={0.0in 0.0in 0.0in 0.0in}, width=0.99\linewidth]{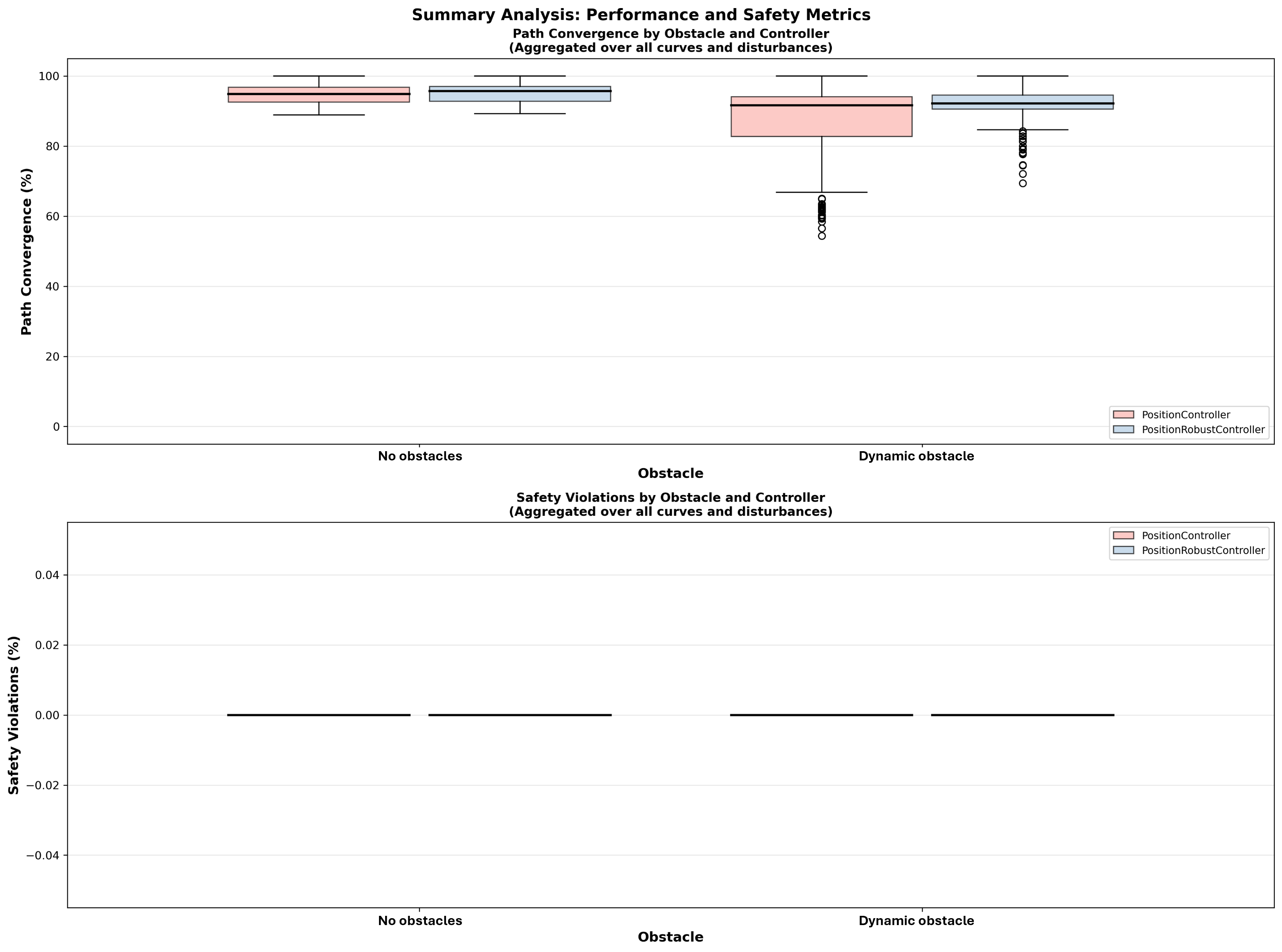}
    \caption{Boxplot of the simulations results showing the percentage of time where the system was inside the tube around the curve (top) and the percentage of the time where the safety constraint was violated (bottom). Results are separated for each obstacle profile and for each controller.}
   \label{fig:sim3dbatch_boxplot}
\end{figure}

\newpage
\section{Conclusion}

This work presented a framework for safe and robust tube-based path-following in robot navigation. The proposed method combines a nominal control with an adaptive robust control law. The nominal control is built using Artificial Vector Field guidance and Backstepping control with reactive smoothed obstacle avoidance. The robust control mitigates the effects of unknown-but-bounded disturbances in a tube around the nominal trajectory. Formal guarantees were established for safety, path convergence, and robustness.



\bibliographystyle{ieeetr} 
\bibliography{bibliography}

\end{document}